\documentclass[12pt]{article}
\usepackage{amsthm,amsfonts,amsmath,amssymb,mathrsfs,url}

\title{Can We Make a Bohmian Electron Reach the Speed of Light,
    at Least for One Instant?}

\author{
Daniel V. Tausk\footnote{Departamento de Matem\'atica,
    Universidade de S\~ao Paulo,
    Rua do Mat\~ao 1010, Cidade Universit\'aria,
    CEP 05508-090 S\~ao Paulo - SP, Brazil.
    E-mail: tausk@ime.usp.br},
Roderich Tumulka\footnote{Department of Mathematics,
     Hill Center, Rutgers University,
     110 Frelinghuysen Road, Piscataway, NJ 08854-8019, USA.
     E-mail: tumulka@math.rutgers.edu}
}
\date{October 25, 2010}

\theoremstyle{plain}\newtheorem{teo}{Theorem}
\theoremstyle{plain}\newtheorem{prop}[teo]{Proposition}
\theoremstyle{plain}\newtheorem{lem}[teo]{Lemma}
\theoremstyle{plain}\newtheorem{cor}[teo]{Corollary}
\theoremstyle{definition}\newtheorem{defin}[teo]{Definition}

\newcommand{\be}{\begin{equation}}
\newcommand{\ee}{\end{equation}}
\newcommand{\Hilbert}{\mathscr{H}}

\newcommand{\RRR}{\mathbb{R}}
\newcommand{\R}{\mathbb{R}}
\newcommand{\dd}{\mathrm d}
\newcommand{\CCC}{\mathbb{C}}
\newcommand{\C}{\mathbb{C}}

\newcommand{\NNN}{\mathbb{N}}
\newcommand{\SSS}{\mathbb{S}}
\newcommand{\D}{\mathcal D}
\newcommand{\E}{\mathcal E}

\newcommand{\vA}{\boldsymbol{A}}
\newcommand{\vQ}{\boldsymbol{Q}}
\newcommand{\vq}{\boldsymbol{q}}
\newcommand{\vv}{\boldsymbol{v}}
\newcommand{\vp}{\boldsymbol{p}}
\newcommand{\vk}{\boldsymbol{k}}
\newcommand{\valpha}{\boldsymbol{\alpha}}

\newcommand{\diag}{\mathrm{diag}}
\newcommand{\Id}{\mathrm I}
\DeclareMathOperator{\Ker}{Ker}

\begin{document}
\maketitle
\begin{abstract}
In Bohmian mechanics, a version of quantum mechanics that ascribes world lines to electrons, we can meaningfully ask
about an electron's instantaneous speed relative to a given inertial frame. Interestingly, according to the relativistic
version of Bohmian mechanics using the Dirac equation, a massive particle's speed is less than or equal to the speed of
light, but not necessarily less. That is, there are situations in which the particle actually reaches the speed of light---a very non-classical behavior. That leads us to the question of whether such situations
can be arranged experimentally. We prove a theorem, Theorem~\ref{thm:Dirac}, implying that for generic initial wave functions the probability that the particle ever reaches the speed of light, even if at only one point in time, is zero. We conclude that the answer to the question is no.

Since a trajectory reaches the speed of light whenever the quantum probability current $\overline{\psi}\gamma^\mu\psi$ is a lightlike 4-vector, our analysis concerns the current vector field of a generic wave function and may thus be of interest also independently of Bohmian mechanics. The fact that the current is never spacelike has been used to argue against the possibility of faster-than-light tunnelling through a barrier, a somewhat similar question.

Theorem~\ref{thm:Dirac}, as well as a more general version provided by Theorem~\ref{thm:E}, are also interesting in their own right.
They concern a certain property of a function $\psi:\RRR^4\to\CCC^4$ that is crucial to the question of reaching the
speed of light, namely being transverse to a certain submanifold of $\CCC^4$ along a given compact subset of
space-time. While it follows from the known transversality theorem of differential topology that this property
is generic among smooth functions $\psi:\RRR^4\to\CCC^4$, Theorem~\ref{thm:Dirac} asserts that it is also generic among
smooth solutions of the Dirac equation.

\medskip

\noindent
 Key words:
 Bohmian mechanics;
 Dirac equation;
 instantaneous velocity;
 quantum probability current;
 generic functions;
 transverse intersections;
 transversality theorem of differential topology.
\end{abstract}

\tableofcontents

\section{Introduction}

In relativistic classical mechanics, the speed of a particle with positive rest mass (such as an electron) can come arbitrarily close to the speed of light $c$ but cannot actually reach it, not even for a single instant, as this would require an infinite amount of energy. In contrast, Bohm's law of motion for a single relativistic electron \cite{Bohm53}, formulated below, prohibits speeds greater than $c$, but not $c$ itself. In other words, there are situations in which a Bohmian electron has speed $c$. As we will see, an infinite amount of energy is not necessary for such situations. Since reaching speed $c$ would be a rather exciting event, the aim of this paper is to look into the conditions under which it occurs.
While this event does not lead to a problem for (or objection against) Bohmian mechanics, it is certainly a feature of Bohmian mechanics that is absent in classical mechanics. The question formulated in the title, i.e., whether these conditions could be arranged by an experimenter, is answered in the negative. We show that, for every external field and for generic initial wave function, the probability of the Bohmian electron to ever reach speed $c$ is zero, and we argue that every preparation procedure should be expected to result in a generic wave function.
Curiously, our arguments also show that a \emph{massless} spin-$\tfrac12$ particle will typically never reach speed $c$, either, in stark contrast to the classical situation; this conclusion also applies to the photon trajectories proposed by Slater \cite{Sla75}.
Our proof is based on a variant of the transversality theorem of differential topology; the main difficulty is to show
that a certain transversality property that is generic among all smooth functions is also generic among the smooth
solutions of the Dirac equation.

\section{Motivation and Relevance}

Bohmian mechanics \cite{Bohm52} is one of the most promising proposals so far for what the reality underlying quantum mechanics might be like; see \cite{DGTZ08,Gol01} for overview articles. While the usual formalism of quantum mechanics describes what observers will see, Bohmian mechanics describes an objective microscopic reality such that observers living in a Bohmian universe will see the statistics described by the quantum formalism. In other words, Bohmian mechanics is a quantum theory without observers and provides, in particular, a solution to the measurement problem of quantum mechanics.

To make clear the question in the title, readers should note that it is not about outcomes of measurements, but about the objective reality as postulated by Bohmian mechanics. Thus, the experiment we are considering is not one in which observers would see anything spectacular, but merely one in which, if Bohmian mechanics is correct, a particular electron has (very probably) reached speed $c$ at some (possibly unknown) time $t$ in the (known) time interval $[t_1,t_2]$. In particular, the experiment would be of no practical use. To anyone \emph{not} interested in the Bohmian version of quantum mechanics, the experiment may have nothing exciting to it. For that reason, some may regard such an experiment as irrelevant; to defend the question, let us give some perspective. To make a comparison, such an experiment would presumably seem more relevant than that of Enders and Nimtz \cite{Nimtz} creating a physical situation of tunnelling in which a certain proposed definition of \emph{tunnelling time} implies a speed inside the barrier greater than $c$; after all, while that definition of tunnelling time is quite arbitrary, Bohmian mechanics is a full-fledged theory that solves all paradoxes of quantum mechanics. (By the way, according to Bohmian mechanics the electron's speed inside the barrier does not exceed $c$.) To make another comparison, the character of the question in the title has similarities to that of the question whether we can, intentionally and artificially, create a black hole that will have closed timelike curves behind the horizon: For both questions, the existence of the spectacular phenomenon is only \emph{inferred} from hypothetical laws of nature, but remains unobservable to us (unless we enter the black hole).

\section{Bohmian Electrons}

A relativistic electron is usually associated with the Dirac equation \cite{DiracEq}
\begin{equation}\label{Dirac}
ic\hbar \gamma^\mu \Bigl(\partial_\mu - i\frac{e}{\hbar} A_\mu\Bigr) \psi = mc^2\psi\,,
\end{equation}
where $\psi:M\to \CCC^4$ is a spinor-valued function on Minkowski space-time $(M,g) = \bigl(\RRR^4, \diag(1,-1,-1,-1)\bigr)$, $x^0=ct$, $m$ is the electron's mass, $e$ its charge, and $A$ a vector field on space-time representing an external electromagnetic potential. According to Bohmian mechanics, in the appropriate version for the Dirac equation \cite{Bohm53}, the electron has a definite position $\vQ(t)$ at every time $t$, corresponding to a world line in space-time $\RRR^4$, and the equation of motion reads
\begin{equation}\label{Bohm}
\frac{d\vQ}{dt} = c\frac{\psi^\dagger \valpha \psi}{\psi^\dagger \psi}(t,\vQ(t))\,,
\end{equation}
where $\valpha = (\alpha^1,\alpha^2,\alpha^3)$ is the triple of the $4\times 4$ Dirac alpha matrices, and
\be
\phi^\dagger \psi = \sum_{s=1}^4 \phi_s^* \psi_s
\ee
is the inner product in $\CCC^4$. Equivalently to \eqref{Bohm}, writing $X=X^\mu$ for the space-time point $(t,\vQ(t))$,
\begin{equation}
\frac{dX^\mu}{d\tau} \propto j^\mu(X) =
\overline{\psi}(X)\, \gamma^\mu \, \psi(X)\,,
\end{equation}
where $\tau$ may be the proper time along the world line or, in fact, any other curve parameter, $j^\mu$ is known as the \textit{quantum probability current vector field} on $M$ associated with $\psi$, $\overline{\psi}=\psi^\dagger\gamma^0$, and $\gamma^\mu$ are the Dirac gamma matrices. Thus, the possible world lines are the integral curves of the vector field $j^\mu$ on $M$.

Since the vector field $j^\mu$ is defined in a covariant way, so are the Bohmian world lines. Moreover, it is known that $j^\mu j_\mu \geq 0$; that is, $j^\mu$ (and thus the world line) is everywhere timelike or lightlike. But it is not true in general that $j^\mu j_\mu >0$; thus, interestingly, the tangent to the world line can be lightlike, in which case the particle \emph{reaches} the speed of light (in every Lorentz frame)---even though the particle has positive rest mass!

According to Bohmian mechanics, the initial position $\vQ(0)$ is random with probability density $\vq\mapsto |\psi(t=0,\vq)|^2$,
thus defining a random world line. It then follows that, also at any other time $t$, the position $\vQ(t)$ has
probability density $\vq\mapsto |\psi(t,\vq)|^2$, and that, furthermore, this is also true in any other Lorentz frame \cite{TT05}.

\section{Spinors Leading to Speed $c$}
\label{sec:spinors}

Given any Lorentz frame, one can read off from \eqref{Bohm} which are the spinors $\psi \in \CCC^4$ for which the velocity $d\vQ/dt$ will point in (say) the $z$ direction with speed $c$: this occurs, and occurs only, for eigenvectors of the matrix $\alpha_z$ with eigenvalue 1. To see this, we use that $\alpha_z$ has eigenvalues $\pm 1$. As a consequence, $\psi^\dagger \alpha_z \psi$ is a real number in the interval $[-\psi^\dagger\psi,\psi^\dagger\psi]$, and the maximal value, which corresponds to speed $c$, is assumed when and only when $\psi$ is an eigenvector of $\alpha_z$ with eigenvalue $1$.

As the speed $|d\vQ/dt|$ cannot exceed $c$, the $x$- and $y$-components of \eqref{Bohm} are zero for eigenvectors of $\alpha_z$. It is known that each of the two eigenspaces $E_{\pm}$ has (complex) dimension 2. Explicitly, using the Dirac representation \cite{Gamma} of the $\alpha$ matrices,
\begin{equation}
\alpha_z = \left( \begin{array}{cc|cc} &&1&\\ &&&-1\\\hline
1&&&\\ &-1&& \end{array} \right)\,,
\end{equation}
the eigenspace $E_+$ for eigenvalue $+1$ is spanned by (the transposes of) $(1,0,1,0)$ and $(0,1,0,-1)$, while the eigenspace  $E_-$ for eigenvalue $-1$ is spanned by (the transposes of) $(1,0,-1,0)$ and $(0,1,0,1)$. Each spinor from $E_{+}$ leads to speed $c$ in the positive $z$ direction, and each from $E_{-}$ to speed $c$ in the negative $z$ direction.

For spinors in $E_+$ to occur it is not necessary that $\psi$ have infinite energy (in whichever sense) or, for that matter, negative energy. For example, think of a superposition of four different plane waves with positive energy. If, at a fixed space-time point $\tilde{x}$, the values of these four basis waves will be four linearly independent elements of $\CCC^4$, then by suitable choice of coefficients we can ensure that $\psi(\tilde{x})$ lies in $E_+$. Explicitly, it is known \cite[p.~9]{DiracEq} that
\be
\psi^{(1)}_{\vk}(x) =
\begin{pmatrix}a_-(\vk)\\0\\-\frac{a_+(\vk) k_3}{|\vk|}\\-\frac{a_+(\vk)(k_1+ik_2)}{|\vk|}\end{pmatrix}
e^{ik_\mu x^\mu}\,, \quad
\psi^{(2)}_{\vk}(x) =
\begin{pmatrix}0\\a_-(\vk)\\-\frac{a_+(\vk) (k_1-ik_2)}{|\vk|}\\\frac{a_+(\vk)k_3}{|\vk|}\end{pmatrix}
e^{ik_\mu x^\mu}
\ee
with
\be
k_\mu = \Bigl(\sqrt{m^2c^2+\hbar^2\vk^2}/\hbar,\vk\Bigr)
\,,\quad
a_\pm(\vk) = \sqrt{1\pm mc/\sqrt{m^2c^2+\hbar^2\vk^2}}
\ee
are two plane wave solutions of the Dirac equation with positive energy $\hbar ck_0$. For any $k>0$, the four functions
\be\label{eq:fourwaves}
\psi^{(1)}_{(0,0,k)}, \psi^{(2)}_{(0,0,k)}, \psi^{(1)}_{(k,0,0)},\psi^{(2)}_{(k,0,0)}
\ee
are linearly independent at any space-time point.

Let $S\subset\CCC^4$ be the set of spinors $\psi\in\CCC^4\setminus\{0\}$ such that the right hand side of Bohm's law \eqref{Bohm} is a vector of length $c$. Consider the partition
\be\label{Somega}
S = \bigcup_{\omega \in \SSS^2} S_\omega\,,
\ee
where $\omega$ is a unit vector in 3-space and $S_\omega$ denotes the set of spinors $\psi\neq 0$ such that the right
hand side of \eqref{Bohm} is $c\omega$. Like for the $z$ direction, it is true for every $\omega$ that $S_\omega\cup\{0\}$ is
a two-dimensional subspace (like $E_+$ above). Hence, in terms of dimensions over the real numbers, $S_\omega$ is a 4-dimensional subspace of the 8-dimensional space $\CCC^4$, and, as we will prove in Section~\ref{sec:proofs}:

\begin{prop}\label{prop:S}
The set $S$ is a $6$-dimensional (embedded) submanifold of $\C^4\cong\R^8$.
\end{prop}

\section{Discussion of the Question in the Title}

Whether we humans (as experimenters) can \emph{make} an electron reach the speed of light, amounts to the question of whether, by preparation of a suitable initial wave function at time $t_1$ and external field $A_\mu$ in the time interval $[t_1,t_2]$, we can arrange a situation in which the Bohmian electron is certain, or at least has significant probability $p$, to reach the speed of light at some instant during the time interval $[t_1,t_2]$. Put differently, $p$ is the probability that the random Bohmian world line passes through a space-time point $x\in [t_1,t_2]\times\RRR^3$ where $j^\mu(x)$ is lightlike, i.e., where $\psi(x)\in S$, or equivalently, where $x$ lies in the set
\be
\Sigma := \bigl\{x\in M: \psi(x) \in S\bigr\}=\psi^{-1}(S)\,.
\ee
Using that $S$ has codimension 2, we will show in Theorem~\ref{thm:Dirac} below that, for generic wave functions $\psi$,
the set $\Sigma$ is a submanifold of space-time $M=\RRR^4$ with codimension 2, i.e., $\dim \Sigma=2$.
However, in that case the probability $p$ that the electron will ever reach the speed of light is zero, as we will
prove in Section~\ref{sec:proofs}:

\begin{prop}\label{prop:zero}
If $\Sigma \subset \RRR^4$ is a submanifold of dimension 2 then the probability that the Bohmian random world line intersects $\Sigma$ is zero.
\end{prop}

For comparison, if $\Sigma$ had codimension 1 (i.e., were a hypersurface), then it could be the case that \emph{every} world line had to intersect $\Sigma$---for example, this is the case if $\Sigma$ is a Cauchy surface. There are such wave functions indeed: for example, $\Sigma$ can be the surface where $t=s$ (in the given Lorentz frame, with a constant $t_1<s<t_2$); to see that such wave functions exist, think of $s$ as the ``initial'' time (for which ``initial'' data of the Dirac equation are specified) and choose the wave function $\psi(s,\vq)$ so that $\psi(s,\vq)\in S$ for every $\vq \in \RRR^3$. Thus, if we could experimentally prepare a wave function of this kind then we would know that the electron's speed $|d\vQ/dt|$ has probability 1 to reach $c$, and indeed must do so at time $s$. However, to prepare a wave function of this kind would require infinite accuracy, as $S$ is a (codimension-2) submanifold, and if the relation $\psi(s,\vq)\in S$ did not hold exactly but merely up to some error (say, the distance of $\psi(s,\vq)$ from $S$ is less than $\varepsilon$) then the speed of the trajectory passing through $(s,\vq)$ would be close but not equal to $c$. It is certainly beyond experimental control whether, for any particular $x\in M$, the relation $\psi(x)\in S$ exactly holds. Therefore, we cannot be confident, for any experimental procedure, that it succeeded in preparing an exceptional wave function for which $p>0$. On the contrary, we should be confident that the procedure prepared a wave function with $p=0$ because they are generic. We conclude that the answer to the question in the title is negative and that, for all wave functions that we can expect to prepare, the probability $p$ of the electron ever reaching speed $c$ is zero.

A further remark concerns the potential $A_\mu$: Just as the initial wave function cannot be experimentally prepared with unlimited precision, the potential cannot either. However, Theorem~\ref{thm:Dirac} implies that already the imprecision in the initial wave function alone suffices to make it impossible to arrange a situation in which $p>0$.

Another remark concerns the global existence of the Bohmian trajectories: A trajectory may fail to exist for all times, and this happens if the trajectory runs
into a zero of the wave function. It is known \cite{TT05} that global existence holds (i.e., the trajectory never hits a zero) for \emph{almost all} initial
positions $\vQ(0)$ at which $\psi(t=0,\cdot)$ is nonzero. Now it follows from our Corollary~\ref{cor:Gdelta} below that generically even more is true:
the wave function is nowhere zero, and thus global existence holds for \emph{all} initial positions.

\section{Genericity Theorem}\label{sec:thm}

The class of wave functions we consider in our genericity theorem, Theorem~\ref{thm:Dirac} below, are those which are,
at the initial time $t_1$, smooth and square-integrable \cite{foot1}. We assume that the external
potential $A_\mu$ is smooth as well. It is then known \cite{Chernoff} that, as a consequence of the Dirac equation,
$\psi$ is a smooth function on space-time, and $\psi(t,\cdot)$ is square-integrable at any time $t$. Moreover, since
the experiment must take place inside some bounded subset of space $C\subset \RRR^3$, it also takes place inside
a compact subset of space-time $K=[t_1,t_2]\times \overline{C}$; thus, it suffices to consider $\Sigma\cap K$
instead of $\Sigma$ \cite{foot2}.

It is a known fact \cite[Chap.~1.3, Thm.~3.3]{Hir76} (indeed, independent of the Dirac equation, requiring only smoothness) that $\Sigma=\psi^{-1}(S)$ will be an (embedded) submanifold of the same codimension as $S$ if $\psi$ is \emph{transverse} to $S$.

\begin{defin}
For any $d\in\NNN$, two linear subspaces of $\RRR^d$ are called \emph{transverse} to each other if and only if they together span $\RRR^d$. Let $M$, $N$ be differentiable manifolds, $S$ a submanifold of $N$ and $f:M\to N$ a smooth map.
We say that $f$ is {\em transverse\/} to $S$ if and only if for all $x\in M$ with $f(x)\in S$, the tangent space to $S$, $T_{f(x)}S\subseteq T_{f(x)}N$, is transverse to the image of the linear mapping $\dd f_x$ that is the first derivative of $f$ in $x$, i.e., if and only if
\begin{equation}\label{eq:deftrans}
\dd f_x(T_xM)+T_{f(x)}S=T_{f(x)}N,
\end{equation}
for all $x\in f^{-1}(S)$. Given a subset $K$ of $M$, we say that $f$ is {\em transverse to $S$ along $K$\/}
if \eqref{eq:deftrans} holds for all $x\in f^{-1}(S)\cap K$.
\end{defin}

We aim at showing that transversality is a generic property. Let us leave aside the Dirac equation for a moment and
consider a generic smooth function $\psi:\RRR^4\to\CCC^4$. Such a $\psi$ is indeed transverse to $S$ according to
the \emph{transversality theorem} of differential topology \cite[Chap.~3.2, Thm.~2.1]{Hir76}. The relevant sense of
``generic'' here and in the following is that the set of $\psi$ that are transverse to $S$ \emph{contains a set that is
open and dense}. The topology used here and in the following is the so-called smooth weak topology, defined as follows.

\begin{defin}
Given an open subset $U$ of $\R^m$, denote by $C^\infty(U,\R^n)$ the space of all smooth maps from $U$ to $\R^n$. The
{\em smooth weak topology\/} of $C^\infty(U,\R^n)$ is the topology of uniform convergence of all partial derivatives over all compact subsets
of $U$; more precisely, it is the topology defined by the following family of semi-norms:
\be\label{eq:seminorms}
\Vert f\Vert_{K,\alpha}=\sup_{x\in K}\Vert\partial_\alpha f(x)\Vert\,,
\ee
where $K$ runs over the compact subsets of $U$ and $\alpha=(\alpha_1,\ldots,\alpha_m)$ runs over the $m$-tuples of non-negative integers.
\end{defin}

Returning to the transversality theorem, it asserts that \textit{if $M$ and $N$ are smooth manifolds and $Q\subseteq N$ is a closed submanifold then the set of smooth functions $\psi:M\to N$ that are transverse to $Q$ along a given compact set $K$ is open and dense in the set $C^\infty(M,N)$ (of all smooth functions $M\to N$) with respect to the smooth weak topology.} (For the definition of the smooth weak topology on manifolds, see \cite[Chap.~2.1]{Hir76}.)

For our case, the transversality theorem implies that the set of smooth functions $\psi:\RRR^4\to\CCC^4$ that are transverse to $S$ along $K$ contains a dense open set.
(Since $S$ is not closed, $\overline{S}=S\cup\{0\}$, this is actually not immediate from the transversality theorem;
it follows, however, directly from Lemma~\ref{thm:lemmaopen} in Section~\ref{sec:proofs} below, which follows from the
transversality theorem.)

Of course, what is relevant to our purpose is not the behavior of functions that are generic among
\emph{all smooth functions} but rather that of functions that are generic among
\emph{all solutions of the Dirac equation}. That is why we formulate our main result as follows.

\begin{teo}\label{thm:Dirac}
Let $\D_A$ be the space of those smooth solutions $\psi$ of the Dirac equation \eqref{Dirac} with potential
$A=A_\mu(x)$ such that $\psi(t,\cdot)$ is square-integrable for every $t\in\RRR$. For every smooth potential $A$ and every
compact subset $K$ of space-time, the subset of functions $\psi\in\D_A$ that have no zeroes in $K$ and are transverse to
$S$ (as defined in Section~\ref{sec:spinors}) along $K$ is open and dense in $\D_A$ relative to the smooth weak topology. In particular, the subset of functions
$\psi\in\D_A$ that are transverse to $S$ along $K$ contains a set that is open and dense with respect to the smooth weak
topology.
\end{teo}

Theorem~\ref{thm:Dirac} is a special case of a more general result, Theorem~\ref{thm:E} below.
Given a subspace $\E$ of $C^\infty(\R^m,\R^n)$ (such as $\D_A$) and a point $x\in\R^m$, we denote by $\E(x)$ the image of $\E$ under the evaluation map $f\mapsto f(x)$, i.e.,
\be
\E(x)=\big\{f(x):f\in\E\big\}.
\ee

\begin{teo}\label{thm:E}
Let $\E$ be a subspace of $C^\infty(\R^m,\R^n)$ and $K$ be a compact subset of $\R^m$ such that $\E(x)=\R^n$, for all $x\in K$.
Assume that $m<n$. If $S$ is a submanifold of $\R^n$ whose closure is contained
in $S\cup\{0\}$ then the set of those $f\in\E$ that have no zeroes in $K$ and that are transverse
to $S$ along $K$ is open and dense in $\E$ with respect to the smooth weak topology.
\end{teo}

From the proof of Theorem~\ref{thm:Dirac}, see Section~\ref{sec:proofs}, it follows also that the restriction to a
compact set $K$ can be omitted at the expense of a slightly weaker (yet widely used) notion of
``generic,'' according to which a property is generic if the set of functions with that property contains
(instead of a dense open set) a dense $G_\delta$ set, i.e., a dense set that is the intersection of countably
many open sets. The corresponding version of Theorem~\ref{thm:Dirac} is formulated as Corollary~\ref{cor:Gdelta} in
Section~\ref{sec:proofs}.

\section{Proofs}
\label{sec:proofs}

We first focus on Proposition~\ref{prop:S}.

\begin{lem}\label{thm:lemembedding}
Let $M$, $N$, $P$ be differentiable manifolds and let $f:M\to N$, $g:N\to P$ be smooth maps. If $g\circ f:M\to P$ is an
embedding then also $f$ is an embedding.
\end{lem}
\begin{proof}
The injectivity of $g\circ f$ implies the injectivity of $f$. The inverse of the bijection $f:M\to f(M)$ is given
by the composition of the map $g\vert_{f(M)}:f(M)\to g\big(f(M)\big)$ with the inverse of the map $g\circ f:M\to g\big(f(M)\big)$
and it is therefore continuous. Finally, given $x\in M$, the injectivity of the linear map
$\dd(g\circ f)(x)=\dd g\big(f(x)\big)\circ\dd f(x)$ implies the injectivity of $\dd f(x)$, so that $f$ is an immersion.
\end{proof}

Recall that if $E$, $F$ are vector bundles over a differentiable manifold $M$ then a {\em vector bundle morphism\/}
from $E$ to $F$ is a smooth map $T:E\to F$ such that, for all $x\in M$, the restriction $T_x$ of $T$ to the fiber
$E_x$ is a linear map taking values in the fiber $F_x$. It is well-known that if $T:E\to F$ is a vector bundle morphism
such that the rank of $T_x$ is independent of $x\in M$ then the kernel $\Ker(T)=\bigcup_{x\in M}\Ker(T_x)$ is
a vector subbundle of $E$ (and, in particular, it is an embedded submanifold of $E$).

\begin{proof}[Proof of Proposition~\ref{prop:S}]
For $\omega\in\R^3$, set $\alpha_\omega=\omega_x\alpha_x+\omega_y\alpha_y+\omega_z\alpha_z$. If $\omega$ is in the
unit sphere $\SSS^2$, the eigenspace corresponding to the eigenvalue $+1$ of $\alpha_\omega$ is the subspace
$S_\omega\cup\{0\}$ of $\C^4$ having complex dimension equal to $2$ (and real dimension equal to $4$).
Consider the trivial vector bundle $\SSS^2\times\C^4$ over the sphere $\SSS^2$ and let $T:\SSS^2\times\C^4\to\SSS^2\times\C^4$
be the vector bundle morphism defined by:
\[T(\omega,\psi)=\big(\omega,(\alpha_\omega-\Id)\psi\big),\quad\omega\in\SSS^2,\ \psi\in\C^4,\]
where $\Id$ denotes the $4\times4$ identity matrix. The kernel $\Ker(T)$ is a vector subbundle $E$ of
$\SSS^2\times\C^4$ whose fibers are $4$-dimensional, so that $E$ is a $6$-dimensional embedded submanifold
of $\SSS^2\times\C^4$ (and of $\R^3\times\C^4$). Let $E'$ denote the complement in $E$ of the zero section of $E$ (so
that $E'$ is an open subset of $E$) and let $f:E'\to\C^4$ denote the restriction to $E'$ of the second projection
of the product $\SSS^2\times\C^4$, so that $f$ is a smooth map and $S$ is the image of $f$. We now use Lemma~\ref{thm:lemembedding}
to establish that $f$ is an embedding. Consider the map $\mathbf v:\C^4\setminus\{0\}\to\R^3$ defined by:
\[
\mathbf v(\psi)=\frac{\psi^\dagger \valpha \psi}{\psi^\dagger \psi},\quad\psi\in\C^4\setminus\{0\},
\]
so that $\mathbf v(\psi)=\omega$, for all $\psi\in S_\omega$. The smooth map $g:\C^4\setminus\{0\}\to\R^3\times\C^4$
defined by $g(\psi)=\big(\mathbf v(\psi),\psi\big)$ has the property that $g\circ f$ is the inclusion map
of $E'$ into $\R^3\times\C^4$, which is an embedding. Hence $f$ is an embedding.
\end{proof}

\begin{proof}[Proof of Proposition~\ref{prop:zero}]
Consider any $x\in\Sigma$. Since $\psi(x)\in S$, $\psi(x)\neq 0$, so that the current vector field $j$ does not vanish
at $x$ and there is a unique nontrivial maximal integral curve of $j$ containing $x$, $C_x$. Since $j$ is everywhere timelike or lightlike, $C_x$ is everywhere timelike or lightlike, and therefore intersects the hyperplane $\{t=t_1\}$ at most once. The set
\be
U=\bigl\{x\in \Sigma: C_x \cap \{t=t_1\}\neq \emptyset\bigr\}
\ee
is open in $\Sigma$, and the mapping $f:U\to \{t=t_1\}$ such that $C_x\cap\{t=t_1\}=\{f(x)\}$ is well defined and, by a standard argument using the implicit function theorem, smooth.
Since the dimension of $\Sigma$ is smaller than the dimension of the hyperplane
$\{t=t_1\}$, $f(U)$ has Lebesgue measure zero. A Bohmian trajectory $t\mapsto(t,\vQ(t))$ intersects $\Sigma$ if and only if it is
$C_x$ for some $x\in U$, and thus if and only if $(t_1,\vQ(t_1))\in f(U)$; since $f(U)$ is a null set, the probability $p$ of that event is
\be
p=\int_{f(U)}\dd^3\vq\, |\psi(t_1,\vq)|^2 = 0\,.
\ee
\end{proof}

We now turn to Theorem~\ref{thm:E} and begin by recalling the following:

\begin{teo}\label{thm:transversality}
Let $M$, $\Lambda$, $N$ be differentiable manifolds, $U$ an open subset of $M\times\Lambda$, $f:U\to N$ a smooth map and $S$ a submanifold of $N$. If $f$ is transverse to $S$ then the set of those $\lambda\in\Lambda$
for which the map $f(\cdot,\lambda)$ is transverse to $S$ contains a countable intersection of open dense subsets of
$\Lambda$ and, in particular, is dense in $\Lambda$.
\end{teo}

\begin{proof}
See \cite[Chap.~3.2, Thm~2.7]{Hir76}.
\end{proof}

If $V$ is a subset of $\R^n$, we denote by $C^\infty(U,V)$ the subset of $C^\infty(U,\R^n)$ consisting of maps with range contained in $V$. By the {\em smooth weak topology\/} of $C^\infty(U,V)$ (or, more generally, of any subset of $C^\infty(U,\R^n)$) we will mean the topology induced from the smooth weak topology of $C^\infty(U,\R^n)$.

\begin{prop}\label{thm:propopen}
Let $U$ be an open subset of $\R^m$, $V$ an open subset of $\R^n$, $S$ a submanifold of $\R^n$ contained
in $V$, and let $K$ be a compact subset of $U$. If $S$ is closed relative to $V$ then the set of those $f\in C^\infty(U,V)$
which are transverse to $S$ along $K$ is open in $C^\infty(U,V)$ with respect to the smooth weak topology.
\end{prop}

\begin{proof}
This is immediate from the transversality theorem, see \cite[Chap.~3.2, Thm~2.1]{Hir76}.
\end{proof}

\begin{lem}
If $\E$ is a subspace of $C^\infty(\R^m,\R^n)$ and $x\in\R^m$ is such that $\E(x)=\R^n$ then
there exists a finite dimensional subspace $\mathcal F$ of $\E$ and a neighborhood $U$ of $x$ such that
$\mathcal F(y)=\R^n$, for all $y\in U$.
\end{lem}

\begin{proof}
Let $f_i\in\E$ be such that $f_i(x)$ equals the $i$-th vector of the canonical basis of $\R^n$. Then
$\big(f_1(y),\ldots,f_n(y)\big)$ is a basis of $\R^n$ for $y$ in some neighborhood $U$ of $x$. Let $\mathcal F$
be spanned by $f_1$, \dots, $f_n$.
\end{proof}

\begin{cor}\label{thm:corfinitesubspace}
If $\E$ is a subspace of $C^\infty(\R^m,\R^n)$ and $K$ is a compact subset of $\R^m$ such that $\E(x)=\R^n$ for all $x\in K$ then there exists a finite
dimensional subspace $\mathcal F$ of $\E$ and an open subset $M$ of $\R^m$ containing $K$ such that $\mathcal F(x)=\R^n$, for all $x\in M$.
\end{cor}

\begin{proof}
For each $x\in K$, let $\mathcal F_x$ be a finite dimensional subspace of $\E$
such that $\mathcal F_x(y)=\R^n$, for all $y$ in an open neighborhood $U_x$ of $x$. By compactness, $K$ is covered by finitely many of the
open sets $U_x$, say $U_{x_1}$, \dots, $U_{x_r}$. Set
$\mathcal F=\sum_{i=1}^r\mathcal F_{x_i}$ and $M=\bigcup_{i=1}^rU_{x_i}$.
\end{proof}

\begin{prop}\label{thm:propdense}
Let $\E$ be a subspace of $C^\infty(\R^m,\R^n)$ and $K$ be a compact subset of $\R^m$ such that $\E(x)=\R^n$, for all $x\in K$.
Assume that $\E$ is endowed with a topology that makes it into a topological vector space (for instance,
the smooth weak topology, being defined by a family of semi-norms, has such property).
If $S$ is a submanifold of $\R^n$ then the set of those $f\in\E$ that
are transverse to $S$ along $K$ is dense in $\E$.
\end{prop}

\begin{proof}
Let $f_0\in\E$ be fixed and let us show that $f_0$ is a limit of maps in $\E$ that are transverse
to $S$ along $K$. Let $M$ and $\mathcal F$ be as in Corollary~\ref{thm:corfinitesubspace}.
Apply Theorem~\ref{thm:transversality} with $\Lambda=\mathcal F$, $N=\R^n$ and
$f:M\times\Lambda\to N$ defined by:
\be
f(x,\lambda)=f_0(x)+\lambda(x),\quad x\in M,\ \lambda\in\Lambda\,.
\ee
The fact that $f_0$ is smooth and that each $\lambda\in\Lambda$ is smooth implies that $f$ is smooth \cite{foot3}.
The derivative of $f$ with respect to the variable $\lambda$ is given simply by evaluation at $x$, so that
the image of $\dd f(x,\lambda)$ contains $\Lambda(x)=\mathcal F(x)=\R^n$. Thus $f$ is a submersion and in particular
transverse to $S$.
It follows that there exists a sequence $(\lambda_i)_{i\ge1}$ in $\Lambda$ converging to zero such that the map:
\be
f(\cdot,\lambda_i)=(f_0+\lambda_i)\vert_M
\ee
is transverse to $S$, for all $i\ge1$; in particular, $f_0+\lambda_i$ is transverse to $S$ along $K$, for all
$i\ge1$. Notice that the convergence of $(\lambda_i)_{i\ge1}$ to zero which we have so far obtained is relative
to the canonical topology of $\Lambda$ as a finite dimensional real vector space (i.e., the topology induced by any linear
isomorphism of $\Lambda$ with $\R^p$). However, since $\E$
is a topological vector space, the inclusion map of $\Lambda$ into $\E$ is continuous, so that $(\lambda_i)_{i\ge1}$ also converges
to zero with respect to the topology of $\E$. Thus $(f_0+\lambda_i)_{i\ge1}$ is a sequence in $\E$
converging to $f_0$ with respect to the topology of $\E$ and we are done.
\end{proof}

\begin{cor}\label{thm:coropendense}
Let $\E$ be a subspace of $C^\infty(\R^m,\R^n)$ and $K$ be a compact subset of $\R^m$ such that $\E(x)=\R^n$, for all $x\in K$.
If $S$ is a submanifold of $\R^n$ which is also a closed subset of $\R^n$
then the set of those $f\in C^\infty(\R^m,\R^n)$ that are transverse to $S$ along $K$ is open and dense with
respect to the smooth weak topology.
\end{cor}

\begin{proof}
Follows directly from Propositions~\ref{thm:propopen} and \ref{thm:propdense}.
\end{proof}

We need a bit more work because in our case, $S$ is not a closed set (its closure is $S\cup\{0\}$, which is not a submanifold because of the cusp at 0).

\begin{lem}\label{thm:lemmaopen}
Let $K$ be a compact subset of $\R^m$ and $S$ be a submanifold of $\R^n$ whose closure is contained in $S\cup\{0\}$.
The set of those $f\in C^\infty(\R^m,\R^n)$ that have no zeroes in $K$ and that
are transverse to $S$ along $K$ is open in $C^\infty(\R^m,\R^n)$ with respect to the smooth weak topology.
\end{lem}

\begin{proof}
Let $f_0\in C^\infty(\R^m,\R^n)$ be a map that has no zeroes in $K$ and that is transverse to $S$ along $K$.
Since $f_0^{-1}\big(\R^n\setminus\{0\}\big)$ is an open subset of $\R^m$ containing the compact set $K$, it contains
the closure of some bounded open set $U$ that contains $K$. Let $\mathfrak V$ denote the subset of $C^\infty(\R^m,\R^n)$
consisting of the maps that have no zeroes in the closure of $U$. The set $\mathfrak V$ is an open neighborhood
of $f_0$ in $C^\infty(\R^m,\R^n)$ with respect to the smooth weak topology \cite{foot4}
and the map:
\begin{equation}\label{eq:restrmap}
\mathfrak V\ni f\longmapsto f\vert_U\in C^\infty\big(U,\R^n\setminus\{0\}\big)
\end{equation}
is continuous, if both $\mathfrak V$ and $C^\infty\big(U,\R^n\setminus\{0\}\big)$ are endowed with the smooth
weak topology. Now, since $S\setminus\{0\}$ is a submanifold of $\R^n$ which is closed in $\R^n\setminus\{0\}$,
Proposition~\ref{thm:propopen} implies that the set $\mathfrak T$ consisting of the maps in $C^\infty\big(U,\R^n\setminus\{0\}\big)$
that are transverse to $S\setminus\{0\}$ along $K$ is open in $C^\infty\big(U,\R^n\setminus\{0\}\big)$.
Hence, the inverse image of $\mathfrak T$ under the map \eqref{eq:restrmap} is an open neighborhood
of $f_0$ in $C^\infty(\R^m,\R^n)$ which consists only of maps that have no zeroes
in $K$ and that are transverse to $S$ along $K$. This concludes the proof.
\end{proof}

\begin{lem}\label{thm:lemmadense}
Let $\E$ be a subspace of $C^\infty(\R^m,\R^n)$ and $K$ be a compact subset of $\R^m$ such that $\E(x)=\R^n$, for all $x\in K$.
Assume that $m<n$. If $S$ is a submanifold of $\R^n$ then the set of those $f\in\E$ that
have no zeroes in $K$ and that are transverse to $S$ along $K$ is dense in $\E$ with respect to the smooth weak
topology.
\end{lem}

\begin{proof}
Since $m<n$, a map $f\in C^\infty(\R^m,\R^n)$ is transverse to the submanifold $\{0\}\subset\R^n$ along $K$ if and only
if $f$ has no zeroes in $K$. Thus, by Corollary~\ref{thm:coropendense}, the set $\mathfrak V$ of those
maps in $\E$ that have no zeroes in $K$ is open and dense in $\E$.
Moreover, by Proposition~\ref{thm:propdense}, the set $\mathfrak T$ of maps in $\E$ that are transverse to $S$ along $K$
is dense in $\E$. Hence the intersection $\mathfrak V\cap\mathfrak T$
is dense in $\E$.
\end{proof}

\begin{proof}[Proof of Theorem~\ref{thm:E}]
Follows directly from Lemmas~\ref{thm:lemmaopen} and \ref{thm:lemmadense}.
\end{proof}

\begin{proof}[Proof of Theorem~\ref{thm:Dirac}]
Set $m=4,n=8$, and $\E=\D_A$. In order to verify the hypothesis $\E(x)=\CCC^4$ for any given space-time point $x=(s,\vq)$, let an arbitrary $\phi\in \CCC^4$ be given and regard $s$ as the ``initial'' time for which we can, according to the standard solution theory of the Dirac equation \cite{Chernoff}, choose the ``initial'' wave function $\psi(s,\cdot)$ arbitrarily (among the smooth and square-integrable functions) and obtain a unique $\psi\in \D_A$ by solving the Dirac equation. It is clear that there are smooth and square-integrable functions $f:\RRR^3\to\CCC^4$ such that $f(\vq)=\phi$.
The conclusion now follows from Theorem~\ref{thm:E}.
\end{proof}

\begin{cor}\label{cor:Gdelta}
Let $A$ be a smooth potential, $\D_A$ be defined as in the statement of Theorem~\ref{thm:Dirac} and $S$
be a submanifold of $\C^4$ whose closure is contained in $S\cup\{0\}$. The subset of functions $\psi\in\D_A$ that have
no zeroes and are transverse to $S$ is a dense $G_\delta$ set in $\D_A$ relative to the smooth weak topology.
In particular, the subset of functions $\psi\in\D_A$ that are transverse to $S$ contains a dense $G_\delta$ set
with respect to the smooth weak topology.
\end{cor}
\begin{proof}
For each $r>0$, let $B_r\subset\R^4$ denote the closed ball of radius $r$ centered at the origin and let
$\mathfrak T_r$ denote the set of maps $\psi\in\D_A$ that have no zeroes in $B_r$ and that are transverse to $S$ along
$B_r$. The set of maps $\psi\in\D_A$ that have no zeroes and are transverse to $S$ equals the intersection
$\mathfrak T=\bigcap_{r=1}^\infty\mathfrak T_r$. Moreover, by Theorem~\ref{thm:Dirac}, each $\mathfrak T_r$ is open
and dense in $\D_A$ with respect to the smooth weak topology, so that $\mathfrak T$ is a $G_\delta$ set. The fact
that $\mathfrak T$ is dense in $\D_A$ would follow if we knew that $\D_A$ is a Baire space. Since this is not clear,
we use the following trick: let $\tau$ be the topology in $\D_A$ defined by the family consisting of the semi-norms \eqref{eq:seminorms}
that define the smooth weak topology {\em and\/} the norm:
\[
\Vert\psi\Vert=\text{$L^2$-norm of the map $\psi(t=0,{\cdot})$}.
\]
It is easy to check that, endowed with $\tau$, the space $\D_A$ is a Frech\'et space (i.e., a Hausdorff
complete topological vector space whose topology is defined by a countable family of semi-norms), so in particular
it is a complete metric space and thus a Baire space. Since the topology $\tau$ is finer than the smooth weak topology,
$\mathfrak T_r$ is also open with respect to $\tau$. Moreover, we can use Proposition~\ref{thm:propdense} to conclude
that $\mathfrak T_r$ is dense in $\D_A$ with respect to $\tau$. Therefore $\mathfrak T$ is dense in $\D_A$ with respect to
$\tau$ and hence also with respect to the smooth weak topology, which is coarser than $\tau$.
\end{proof}

\section{Variations of the Question in the Title}

A natural further question involves conditionalizing on the result of an observation: Can we arrange a situation between time $t_1$ and $t_2$ (i.e., prepare an initial wave function at $t_1$ and create an external field $A_\mu$ between $t_1$ and $t_2$) such that, if a certain measurement at time $t_2$ comes out favorably, probability was high that the electron reached speed $c$ at some time between $t_1$ and $t_2$? Again, the answer is negative. To be sure, we could try measure the electron's position $\vQ(t_2)$ and consider those outcomes as favorable that correspond (as we can compute if we know $\psi(t_1,\cdot)$ and $A_\mu$) to world lines that have intersected $\Sigma$; then favorable outcomes would indeed imply that the electron has passed through $\Sigma$, yet favorable outcomes would have probability zero. More generally, if the electron has (unconditional) probability zero to pass through $\Sigma$, then conditionalizing on any event of nonzero probability will not change that.

It may be reasonable to consider only initial wave functions without contributions of negative energy, i.e., $\psi(t_1,\cdot)\in \Hilbert_+$ with
$\Hilbert_+\subset L^2(\RRR^3,\CCC^4)$ the subspace of positive energy for the Hamiltonian $H_0$ corresponding to the Dirac equation with $A_\mu=0$.
Among such initial data it is again generic that $\psi$ is transverse to $S$ (and thus $p=0$) along any compact subset $K\subset \RRR^4$. This follows from
Theorem~\ref{thm:E} in the same way as Theorem~\ref{thm:Dirac} by taking $\E$ to be the space of smooth solutions $\psi$ of the Dirac equation with
$\psi(t_1,\cdot)\in\Hilbert_+$. The fact that $\E(x)=\CCC^4$ for all $x\in\RRR^4$ is proven as follows: first, observe that using
linear combinations of the four plane waves \eqref{eq:fourwaves} one obtains solutions $\psi$ without contributions of negative energy and having an arbitrary
given value at a given point $x\in\RRR^4$. Such solutions $\psi$ are not square integrable, but we can take care of this using the following trick:
the plane waves \eqref{eq:fourwaves}
are Fourier transforms of delta functions. By considering Fourier transforms of continuous compactly supported approximations of delta functions, one sees
that the plane waves \eqref{eq:fourwaves} are pointwise limits of elements of $\E$. This proves that $\E(x)$ is a dense subspace of $\CCC^4$ and hence that
$\E(x)=\CCC^4$.

Another remark concerns the mass $m$ that appears in the Dirac equation. It played no role in our reasoning whether $m>0$ or $m=0$. As a consequence, also for a massless spin-$\tfrac12$ particle (which apparently does not exist in nature), governed by the Dirac equation with $m=0$, the Bohmian trajectory generically has probability zero to ever reach speed $c$. This conclusion is perhaps surprising because classical massless particles always move at speed $c$, and because the plane wave solutions to the massless Dirac equation always have wave vectors $k^\mu$ which are lightlike (while for $m>0$ they are timelike), so that, if $\psi$ is a plane wave, the Bohmian particle always moves at speed $c$. Of course, plane waves are very exceptional functions.

Ironically, it also follows from Theorem~\ref{thm:E}, in very much the same way as for a spin-$\tfrac12$ particle, that Slater's \emph{photon} trajectories \cite{Sla75}, defined by taking the wave function to be mathematically equivalent to a classical Maxwell field $(\boldsymbol{E},\boldsymbol{B})$, replacing the Dirac equation with the source-free Maxwell equations, and taking the probability current to be $j^\mu=\bigl(\tfrac12(\boldsymbol{E}^2+\boldsymbol{B}^2),\boldsymbol{E}\times\boldsymbol{B}\bigr)$, generically have probability zero to ever reach the speed $c$ that we call the speed of light.

Also the question arises whether there exist any wave functions $\psi$ at all, even if exceptional, such that $\Sigma=\psi^{-1}(S)$ is 4-dimensional, i.e., contains an open set---so that there is positive probability the electron spends a
positive amount of time at speed $c$. In case $m=0$ the answer is obviously yes, as then the plane waves are such that $\psi(\RRR^4) \subseteq S$. For $m>0$, here is an example with $\psi(\RRR^4)\subset S$:
\be\label{example}
\psi(t,\vq) =
\cos \omega t \begin{pmatrix}1\\1\\1\\-1\end{pmatrix}
+\sin\omega t \begin{pmatrix}-i\\-i\\+i\\-i\end{pmatrix}
\ee
with $\omega = mc^2/\hbar$. This $\psi$ solves the Dirac equation (in the Dirac representation \cite{Gamma}) with $A_\mu(x)=0$ and yields the current
\be
j = 4\bigl(1,0,-\sin 2\omega t,\cos 2\omega t\bigr)\,,
\ee
which is lightlike. The Bohmian trajectory is:
\be
\vQ(t) = \text{const.} + \frac{c}{2\omega} (0,\cos 2\omega t, \sin 2\omega t)\,,
\ee
i.e., the electron moves along a circle in the $yz$-plane at speed $c$. (The example \eqref{example} is not square-integrable, but we can easily turn it into a square-integrable example by choosing the initial wave function at $t=0$ so that it coincides with \eqref{example} at $t=0$ in the ball of radius $R>0$ around the origin in $\RRR^3$ (and outside so that it is smooth and square-integrable), and evolving it with $A_\mu(x)=0$. By the finite propagation speed of the Dirac equation, the wave function will coincide with \eqref{example} on the ``diamond-shaped'' open set
\be
\bigl\{(t,\vq)\in\RRR^4:-R/c<t<R/c,|\vq|<R-c|t|\bigr\}\,,
\ee
which is therefore contained in $\Sigma$, as desired.)

A last variant of our question concerns a system of $N>1$ entangled electrons: Can we prepare an initial wave function $\psi(t_1,\cdot):(\RRR^3)^N\to(\CCC^4)^{\otimes N}$ and create a field $A_\mu$ between $t_1$ and $t_2$ such that, with nonzero probability, at least one of the $N$ electrons will reach speed $c$ at some time between $t_1$ and $t_2$? Here, we refer to the following $N$-particle version of Bohmian mechanics: Let $\psi$ evolve according to the Hamiltonian $H=H_1+\ldots+H_N$, where each $H_k$ is a copy of the Hamiltonian of the Dirac equation acting on the $k$-th particle; that is, the $N$ electrons do not interact. Furthermore, let $\vQ_k(t)$ denote the position of the $k$-th particle at time $t$, and let it move according to
\begin{equation}\label{NBohm}
\frac{d\vQ_k}{dt} = c\frac{\psi^\dagger \valpha_k \psi}{\psi^\dagger \psi}
\bigl(t,\vQ_1(t),\ldots,\vQ_N(t)\bigr)\,,
\end{equation}
where $\valpha_k$ is $\valpha$ acting on the spinor index of the $k$-th particle. The wave function must be anti-symmetrized because electrons are fermions.

Again, the answer is negative. This can be proved from Theorem~\ref{thm:E} in much the same way as for $N=1$ after noting that the $k$-th particle reaches the speed of light, $|d\vQ_k/dt|=c$, if and only if
\be\label{Nc}
\psi(t,\vQ_1(t),\ldots,\vQ_N(t)) \in \tilde{S}_k
= \biggl(\,\bigcup_{\omega\in \SSS^2}
(\CCC^4)^{\otimes(k-1)}\otimes (S_\omega\cup\{0\}) \otimes(\CCC^4)^{\otimes (N-k)}\biggr)\setminus\{0\}\,,
\ee
where $S_\omega$ is, as in \eqref{Somega}, the set of spinors in $\CCC^4\setminus\{0\}$ such that the associated
Bohmian velocity is $c\omega$; recall that $S_\omega\cup\{0\}$ is a complex subspace of $\CCC^4$ of dimension 2.
The same arguments as before show that the set $\tilde{S}_k$ is a real (embedded) submanifold of codimension $4^N-2$
and that generically $\psi$ will be transverse to $\tilde{S}_k$ for each $k=1,\ldots, N$ along any compact set $K\subset \RRR\times \RRR^{3N}$ that does not
intersect the diagonal
\[\RRR\times \{(\vq_1,\ldots,\vq_N): \vq_i=\vq_j \text{ for some }i\neq j\}.\]
For $N\ge2$, since $4^N-2$ is larger than the dimension $3N+1$ of the domain of $\psi$, it follows that for generic
$\psi$ the sets $\psi^{-1}(\tilde{S}_k)\cap K$, $k=1,\ldots,N$ are empty. Furthermore, we believe that the negative answer does not depend on
whether or not we allow $\psi$ to have contributions of negative energy.

\section{Comparison with Orthodox Quantum Mechanics}

In orthodox quantum mechanics, it is not possible to ask a question corresponding to the one in the title. Let us explain. One may think of a \emph{quantum measurement} corresponding to an observable representing instantaneous velocity. (Arguably, this observable would be, in non-relativistic quantum mechanics, $\hat\vv=\hat\vp/m$, i.e., the ``momentum observable'' $\hat\vp=-i\hbar\nabla-e \vA$ with $A_\mu = (A_0,\vA)$ divided by the mass $m$ of the particle; in relativistic quantum mechanics it would be $\hat \vv = c^2 \hat\vp\hat H^{-1}$ with $\hat H$ the Hamiltonian, provided that each component of $\hat\vp$ commutes with $\hat H^{-1}$.) In practice, though, experiments that ``measure the momentum observable'' are carried out by measuring the position several times. Concretely, e.g., the position could be measured at times $t$ and $t'$, with results $\vq$ and $\vq'$ and respective inaccuracies $\delta q$ and $\delta q'$. The velocity would then be inferred to be
\be\label{vdef}
\vv=(\vq'-\vq)/(t'-t)
\ee
with inaccuracy $(\delta q'+\delta q)/(t'-t)$. But clearly, that vector $\vv$ must not be regarded as an \emph{instantaneous velocity} but rather as a time-averaged velocity. After all, in any theory in which particles have trajectories, and in which $\vq$ and $\vq'$ are the exact positions at times $t$ and $t'$, respectively, \eqref{vdef} would define the average velocity in the time interval $[t,t']$. As a consequence, even if speed $c$ were reached for one instant but not more, $|\vv|$ would be strictly less than $c$, and knowing $\vv$ would permit no conclusion about whether or not $c$ was reached.

Moreover, even if somehow there existed an experiment that could be regarded as a quantum measurement of instantaneous velocity, this would only permit us to ask
\be\label{qt}
\text{for one specific time $t$ whether speed}(t)=c,
\ee
but not
\be\label{qet}
\text{whether $\exists t\in[t_1,t_2]$ such that speed}(t)=c.
\ee
In the Bohmian framework, question \eqref{qet} can be asked and indeed is more relevant, as there are situations in which the probability of
the event speed$(t)=c$ vanishes for every $t$ while with nonzero probability there is $t\in [t_1,t_2]$ with speed$(t)=c$ (e.g., when $\Sigma$ is a timelike hypersurface). However, question \eqref{qet} cannot be asked in the orthodox framework, as quantum observables are not regarded as having values except when they are measured.

As a last remark, readers may be puzzled by the impression that the following three facts together lead to a contradiction: (i)~Bohmian mechanics makes the same testable predictions as standard quantum mechanics; (ii)~Bohmian mechanics implies that an electron can reach the speed of light; (iii)~the spectrum of the momentum four-vector observable of a relativistic electron is a ``mass shell'' hyperboloid, and in particular is disjoint from the light cone. Here is why no contradiction arises: What the experiments that are usually called ``quantum measurements of the momentum observable'' actually measure, in a world governed by Bohmian mechanics, is mass times the \emph{asymptotic velocity} $\lim_{t\to\infty} (\vQ(t)-\vQ(0))/t$, which is different from the \emph{instantaneous velocity} $\lim_{t\to 0} (\vQ(t)-\vQ(0))/t$.
The latter, but not the former, reaches $c$. (An example is provided by the wave function \eqref{example}: If we ignore that it is not square integrable then standard quantum mechanics predicts that a ``quantum measurement of the 3-momentum observable'' will yield 0 with certainty,
and in fact $\lim_{t\to\infty} (\vQ(t)-\vQ(0))/t=0$ for every trajectory, while the instantaneous speed is always $c$.)

\bigskip

\textit{Acknowledgments.}
R.T.\ thanks Frank Loose (T\"ubingen) and Feng Luo (Rutgers) for helpful discussions. This research was supported by grant RFP1-06-27 from The Foundational Questions Institute (fqxi.org).


\begin{thebibliography}{19}
\bibitem{foot1} Readers may worry whether it is appropriate that
we consider only smooth wave functions, as they may expect that wave functions might generically be non-smooth. However,
there are reasons to believe that wave functions in the real world actually are smooth: First, in Bohmian mechanics,
if the wave function of the universe is smooth then the wave function of every subsystem, i.e., its conditional wave
function, is smooth, too. Second, perhaps surprisingly, for relevant statistical ensembles of wave functions
(those connected to thermal equilibrium), smoothness is actually generic \cite{TZ}.
\bibitem{foot2} Indeed, it would even be reasonable to assume that the initial wave function
$\psi(t_1,\cdot)$ has compact support. Since the speed of propagation of disturbances in the Dirac equation is
limited by $c$, the support of $\psi(t,\cdot)$ cannot grow faster than at speed $c$, and thus is compact at any
time $t$. As a further consequence, if the support of the initial wave function is contained in a ball of radius
$R$ and if $C$ is a ball with the same center and radius $R+(t_2-t_1)c$ then the electron's trajectory cannot leave $K$ during the time interval $[t_1,t_2]$. For our genericity theorem, however, it is not relevant whether the initial wave function is compactly supported.
\bibitem{foot3} In order to see that, identify $\Lambda$ with $\R^p$ using a basis $(\lambda_1,\ldots,\lambda_p)$ and notice
that $f$ is identified with the map $(x,c_1,\ldots,c_p)\mapsto f_0(x)+\sum_{i=1}^pc_i\lambda_i(x)$.
\bibitem{foot4} Actually, $\mathfrak V$ is even open with respect to the compact open topology (i.e., the topology of uniform convergence over compact
subsets), which is coarser than the smooth weak topology.
\end{thebibliography}
\end{document}